\documentclass[12pt]{article}

\usepackage[colorlinks,citecolor=blue,urlcolor=blue]{hyperref}
\usepackage{amsthm,amsmath,amsfonts,amssymb,bbm}
\usepackage{accents}
\usepackage{algorithm}
\usepackage{algpseudocode}
\usepackage{graphicx}
\usepackage{mathtools}
\usepackage{subcaption}
\usepackage{tikz}

\usepackage{tabularray}
\usepackage{booktabs}
\usepackage{multirow}
\usepackage{hhline}
\usepackage{makecell}
\usepackage{tcolorbox}

\DeclarePairedDelimiter\floor{\lfloor}{\rfloor}

\DeclareMathOperator*{\argmin}{arg\,min}

\let\originalleft\left
\let\originalright\right
\renewcommand{\left}{\mathopen{}\mathclose\bgroup\originalleft}
\renewcommand{\right}{\aftergroup\egroup\originalright}

\newcommand{\norm}[1]{\left\lVert#1\right\rVert}

\newcommand{\expect}[2]{\mathbb{E}_{#1}\left[#2\right]}
\newcommand{\expec}[1]{\mathbb{E}\left[#1\right]}

\newcommand{\Var}[1]{\textrm{Var}\left[#1\right]}

\renewcommand{\d}{\text{d}}
\newcommand{\simiid}{\stackrel{iid}{\sim}}

\newcommand{\R}{\mathbb{R}}

\newcommand{\N}{\mathbb{N}}

\newcommand{\abs}[1]{\left|#1\right|}

\newcommand{\eset}[1]{{\left\{#1\right\}}}

\renewcommand{\t}[1]{{\textrm{#1}}}

\newcommand{\Prob}[1]{\mathbb{P}\left[#1\right]}

\newtheoremstyle{def}
{12pt}   
{6pt}   
{\normalfont}  
{0pt}       
{\bfseries} 
{.}         
{5pt plus 1pt minus 1pt} 
{}          

\theoremstyle{plain}
\newtheorem{theorem}{Theorem}[section]
\newtheorem{lemma}{Lemma}[section]

\newtheorem{example}{Example}[section]
\newtheorem{assumption}{Assumption}[section]
\newtheorem{definition}{Definition}[section]
\newtheorem{proposition}{Proposition}[section]

\graphicspath{{figures/}}

\def\centerarc[#1](#2)(#3:#4:#5)
{ \draw[#1] ($(#2)+({#5*cos(#3)},{#5*sin(#3)})$) arc (#3:#4:#5); }

\title{Isotropic randomization for one-sample testing in metric spaces}
\author{Matthieu Bult\'e$^{1,2}$ and Helle Sørensen$^1$}
\date{%
  {\small
    $^1$Department of Mathematical Sciences, University of Copenhagen\\%
    $^2$Faculty of Business Administration and Economics, Bielefeld University\\[2ex]%
  }
}

\begin{document}

\maketitle

\begin{abstract}
  We address the problem of testing hypotheses about a specific value of the Fréchet mean in metric spaces, extending classical mean testing from Euclidean spaces to more general settings. We extend an Euclidean testing procedure progresively, starting with test construction in Riemannian manifolds, leveraging their natural geometric structure through exponential and logarithm maps, and then extend to general metric spaces through the introduction of admissible randomization techniques. This approach preserves essential geometric properties required for valid statistical inference while maintaining broad applicability. We establish theoretical guarantees for our testing procedure and demonstrate its effectiveness through numerical experiments across different metric spaces and distributional settings. The practical utility of our method is further illustrated through an application to wind data in western Denmark, showcasing its relevance for real-world statistical analysis.
\end{abstract}

\section{Introduction} \label{sec-intro}

The statistical analysis of non-standard data types has gained increasing attention as new methods of measurement and data collection emerge across various fields. This has led to the development of methods for analyzing random variables taking values in metric spaces, also called random objects, where only a notion of distance between points is available rather than the rich structure of a vector space.

The study of random objects spans multiple application domains. In functional data analysis, methods have been developed for analyzing curve data \cite{ramsay_functional_2005}. Random objects also appear in neuroimaging through the analysis of correlation matrices from fMRI data \cite{petersen_frechet_2019, chewi_gradient_2020, thanwerdas_theoretically_2022}, and in network science through the study of adjacency matrices representing social networks \cite{dubey_frechet_2020}.

Probability distributions are a particularly well-studied example of random objects and various approaches have been developed for their analysis. They have been studied as images of Hilbert spaces under transformations \cite{petersen_functional_2016}, as specific Hilbert spaces with tailored addition and scalar multiplication operators \cite{van_den_boogaart_bayes_2014}, and as metric spaces with distances constructed to expose certain properties or invariances \cite{panaretos_invitation_2020, srivastava_functional_2016}.

The statistical theory for random objects has seen substantial development in recent years. Fundamental work has addressed hypothesis testing and inference \cite{dubey_frechet_2019, dubey_frechet_2020, mccormack_stein_2022, kostenberger_robust_2023}, alongside various approaches to regression \cite{petersen_frechet_2019, bulte_medoid_2024, hanneke_universal_2021} and time series models \cite{jiang_testing_2023, bulte_autoregressive_2024}. Since metric spaces offer limited inherent structure, additional assumptions are often introduced to ensure well-defined statistical quantities. A common approach is to assume the metric space is a Hadamard space, which provides a rich geometric framework while maintaining generality \cite{auscher_probability_2003, bacak_computing_2014}.

The Fr\'echet mean, a generalization of the expected value to metric spaces, has been the subject of extensive theoretical investigation. It exposes unexpected asymptotic properties which are not found in Euclidean spaces, such as for instance the example of smeariness \cite{eltzner_smeary_2018}. Recent work has examined its concentration properties \cite{brunel_concentration_2023}, asymptotic distributional behavior \cite{bhattacharya_omnibus_2017, yokota_law_2018}, or more fundamental properties of the quantity itself in various scenarios \cite{mccormack_equivariant_2023, mccormack_stein_2022, hundrieser_lower_2024}.

In this paper, we focus on the problem of testing whether the Fr\'echet mean of a distribution on a metric space equals a hypothesized element of the space. This extends the classical problem of testing the mean in Euclidean space to the more general setting of metric spaces. We start by considering the case of Riemannian manifolds, where the exponential and logarithm maps provide natural tools for constructing tests. We then extend our framework to general metric spaces through the introduction of admissible randomization, which preserve key geometric properties needed for valid inference. We demonstrate the practical utility of our approach through numerical experiments and a case study of circular data.

The paper is organized as follows: Section \ref{sec-background} provides an introduction to Fréchet means in metric spaces and to the mean testing problem. It presents a gradually more general solution to the mean testing problem up to an approach in Riemannian manifolds. Section \ref{sec-meantest} introduces the main contribution of the paper, a general approach to mean testing in metric spaces. Section \ref{sec-numerical} illustrates the performance of the test on a series of numerical experiments in various metric spaces and distributional setups. Finally, Section \ref{sec-realdata} demonstrates the use of the method to a real dataset.
\section{Background} \label{sec-background}

\subsection{Fréchet mean and metric spaces}
Let $(\Omega, d)$ be a metric space equipped with the Borel $\sigma-$algebra induced from the metric topology on $\Omega$. A random variable $X$ over $\Omega$ is a Borel measurable function from some probability space to $\Omega$. For $p \geq 1$, the space $L^p(\Omega)$ contains all random variables $X$ such that the $p-$th moment of the distance function is well-defined, that is, $\expec{d(X, \omega)^p} < \infty$ for some $\omega \in \Omega$ --- and hence for all $\omega' \in \Omega$ since by the triangle inequality and Jensen's inequality $\expec{d(\omega', X)^p} \leq 2^{p-1}\expec{d(\omega, X)^p} + 2^{p-1} d(\omega, \omega')^p < \infty$. For a random variable $X \in L^2(\Omega)$, consider the expected value
\begin{equation}\label{eq-frechet-fn-def}
  F_X(\omega) = \expec{d(X, \omega)^2}.
\end{equation}
Fréchet \cite{frechet_elements_1948} proposes the minimizer of this quantity as a generalization of the expectation in the Euclidean case. This minimizer, when it exists, together with the minimal value attained by $F_X$, are commonly called the \textit{Fréchet mean} and \textit{Fréchet variance} of $X$
\begin{equation}\label{eq-frechet-def}
  \expec{X} = \argmin_{\omega \in \Omega} F_X(\omega) \qquad \Var{X} = F_X(\expec{X}).
\end{equation}
The Fréchet mean generalizes the expected value to metric spaces, and, similarily to Euclidean spaces, provides a notion of \textit{center of the distribution} of $X$. This stems from the fact that for a real-valued random variable $X \sim P$, the integral $\int x P(\d x)$ is the minimizer of the Fréchet function. However, unlike the Euclidean case, the Fréchet mean is not guaranteed to exist for any random variable in $L^2(\Omega)$, and when it does, it is not necessarily unique. Consider for example the case where the space $\Omega$ is the $d-$dimensional sphere $S^d = \eset{ x \in  \R^{d+1} : \norm{x} = 1 }$ and $X$ is uniformly distributed on $S^d$. Then $F_X(\omega)$ is constant and hence a unique minimizer does not exist. 

Given a sample $X_1, \ldots, X_n$ independent of copies of $X \in L^2(\Omega)$ the Fréchet mean and variance of $X$ can be estimated from their sample counterparts constructed through minimization of the empirical Fréchet function,
\begin{equation}\label{eq-frechet-emp}
  \hat \mu_n = \argmin_{\omega \in \Omega} \frac{1}{n} \sum_{i=1}^n d(X_i, \omega)^2 \qquad \hat V_n = \frac{1}{n} \sum_{i=1}^n d(X_i, \hat\mu_n)^2.
\end{equation}

The following assumptions will be made throughout the paper to quantity the behavior of these estimators. The first assumption, common in the study of Fréchet means and more generally M-estimators, requires that the theoretical and empirical Fréchet means uniquely exist. It is a central assumption in proving that the empirical Fréchet mean is consistent, that is, $d(\mu, \hat\mu_n) = o_P(1)$, see for instance \cite[Corrolary 3.2.3]{vaart_asymptotic_1998}.

\begin{assumption} \label{ass-separated}
  The random variable $X$ has a unique Fréchet mean $\mu \in \Omega$. Its sample estimator $\hat\mu_n$ exists almost surely, and for any $\varepsilon > 0$, the population Fréchet mean satisfies $\inf_{d(\omega, \mu) > \varepsilon} F_X(\omega) > F_X(\mu)$.
\end{assumption}

The second assumption, also standard in the study of M-estimators, provides control of the complexity of the metric space, see for instance \cite{vaart_asymptotic_1998}. It is commonly found in various forms in the study of random objects and Fréchet means, see \cite{dubey_frechet_2019, schotz_convergence_2019-1}.

\begin{assumption} \label{ass-cover}
  Let $N(\varepsilon, U)$ be the covering number of $U \subset \Omega$ with balls of size $\varepsilon$. Assume the following
  \begin{enumerate}
    \item For any $\omega \in \Omega$, $\int_0^1 \sqrt{1 + \log N(\varepsilon\delta/2, B_\delta(\mu))}\, \d \varepsilon \rightarrow 0$ as $\delta \rightarrow 0$.
    \item The entropy integral $\int_0^1 \sqrt{1 + \log N(\varepsilon, \Omega)}\, \d \varepsilon$ is finite.
  \end{enumerate}
\end{assumption}

We will be concerned with the problem of using a sample $X_1, \ldots, X_n$ independent of copies of $X \in L^2(\Omega)$ to test whether the Fréchet mean of $X$ takes a specific hypothesized value $\mu_0 \in \Omega$,
\begin{equation}\label{eq-hypothesis}
  H_0: \expec{X} = \mu \qquad\t{vs.}\qquad H_1: \expec{X} \neq \mu.
\end{equation}
Our test is based on the empirical Fréchet mean and variance of $X$, as well as their behavior under the null hypothesis, based on the results in \cite{dubey_frechet_2019}. The contribution of this work lies in proposing an approach to randomization suitable for the testing problem in general metric spaces and finite sample. 


\subsection{Mean test on the real line}\label{subsec-realline}

Suppose that we observe independent and identically distributed random variables $X_1,\ldots,X_n \sim P$ in $\R$ with $\expec{\abs{X}} < \infty$. Given a $\mu \in \R$, we aim at a testing procedure for the hypothesis in Equation \eqref{eq-hypothesis} with no parametric assumptions on $P$. For simplicity, we restrict the model class to distributions symmetric around $\mu$, which is equivalent to saying that the law of $X$ is invariant under the \textit{reflection map} $g_\mu : x \mapsto 2\mu - x$,
\begin{equation*}
  \Prob{X \in A} = \Prob{g_\mu \cdot X \in A} \qquad \forall A \in \mathcal{B}(\R),
\end{equation*}
where we use the notation $g_\mu \cdot X$ to denote the function application $g_\mu(X)$. Following the approach presented in Chapter 17 of \cite{lehmann_testing_2022}, the symmetry of $X$ can be used to construct a randomization that preserves the distribution of $X$ under the null hypothesis. Given a $z \in \eset{0, 1}$, define the \textit{randomized variable}
\begin{equation}\label{eq-r-randomization}
  g_\mu^z \cdot X =
  \begin{cases}
    X &\t{if } z = 0, \\
    g_\mu \cdot X &\t{if } z = 1.
  \end{cases}
\end{equation}
By symmetry, under the null hypothesis, $g_\mu^z \cdot X$ has the same distribution as $X$ for $z \in \eset{0, 1}$. This property carries over to the distribution of any test statistic $T$ evaluated on a randomized sample. That is, under the assumption of symmetry and an arbitrary binary vector $z \in \eset{0, 1}^n$, the randomized statistic $T(g_\mu^{z_1} \cdot X_1, \ldots, g_\mu^{z_n} \cdot X_n)$ has the same distribution as $T(X_1, \ldots, X_n)$.

Consider now all $2^n$ binary randomization vectors and let $T_{(1)}, \ldots, T_{(2^n)}$ be the associated evaluations of the test statistic, sorted. For a nominal level $\alpha \in (0, 1)$, a level-$\alpha$ hypothesis test for $H_0 : \expec{X} = \mu$ can be constructed by rejecting the null hypothesis if the observed test statistic $T(X_1, \ldots, X_n)$ is \textit{too extreme} compared to the randomized sample. Since the randomization $g_\mu$ is self-inverse, the set $\eset{\t{id}, g_\mu}$ is a group, and the randomization procedure, together with the symmetry of $X$, corresponds to the \textit{Randomization Hypothesis} framework of \cite{lehmann_testing_2022}. By Theorem 17.2.1 of the same manuscript, this test has the desired level $\alpha$.

In practice though, applying the $2^n$ randomizations to construct the test is computationally infeasible for large $n$. Instead, one recognized that the testing procedure can be linked to an expected value with respect to the uniform distribution over the product group $\eset{\t{id}, g_\mu}^n$, which can be approximated via Monte Carlo techniques by sampling a large number of randomizations and computing the test statistic for each of them. For a given number of replicates $B$ and a significance level $\alpha$, one can sample $Z_1, \ldots, Z_n \simiid \t{Bernoulli}(1/2)$ and compute the test statistic $T_b = T(X_1^\star, \ldots, X_n^\star)$ where $X_i^\star = g_\mu^{Z_i} \cdot X_i$. Then, the level-$\alpha$ hypothesis test is constructed following the same procedure as outlined above, applied to the sorted statistics $T_{(1)}, \ldots, T_{(B)}$. A natural choice for the test statistic $T(X_1, \ldots, X_n)$ is the empirical variance of the sample, given by
\begin{equation}\label{eq-var-emp}
  \hat V_n(X_1, \ldots, X_n) = \frac{1}{n} \sum_{i=1}^n (X_i - \bar X_n)^2.
\end{equation}

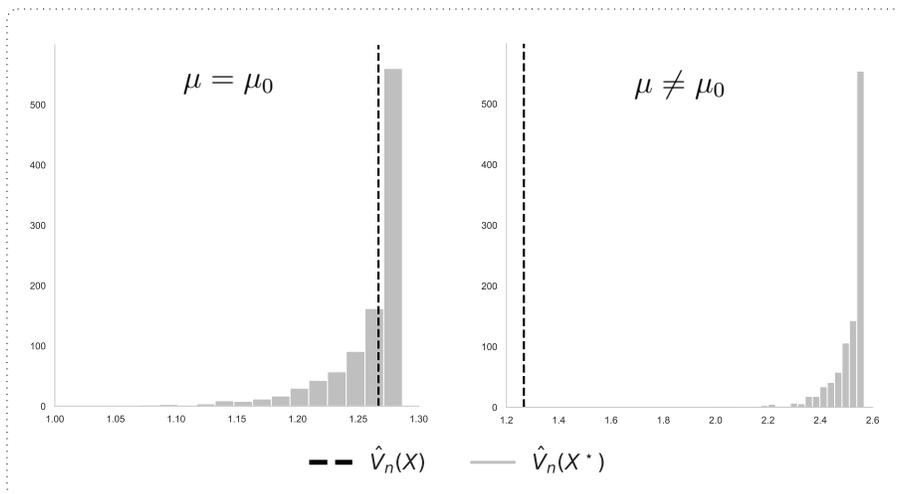
\begin{figure}[!t]
    \hspace*{-.8cm}
    \centering
    \begin{tikzpicture}
        \draw[dotted, rounded corners=4pt] (0, 0) -- ++(12, 0) -- ++(0,-6.5) -- ++(-12, 0) -- cycle;
        \node at (6+2.95, -1) {$\mu \neq \mu_0$};
        \node at (6+2.95, -3) {\includegraphics[width=5.5cm]{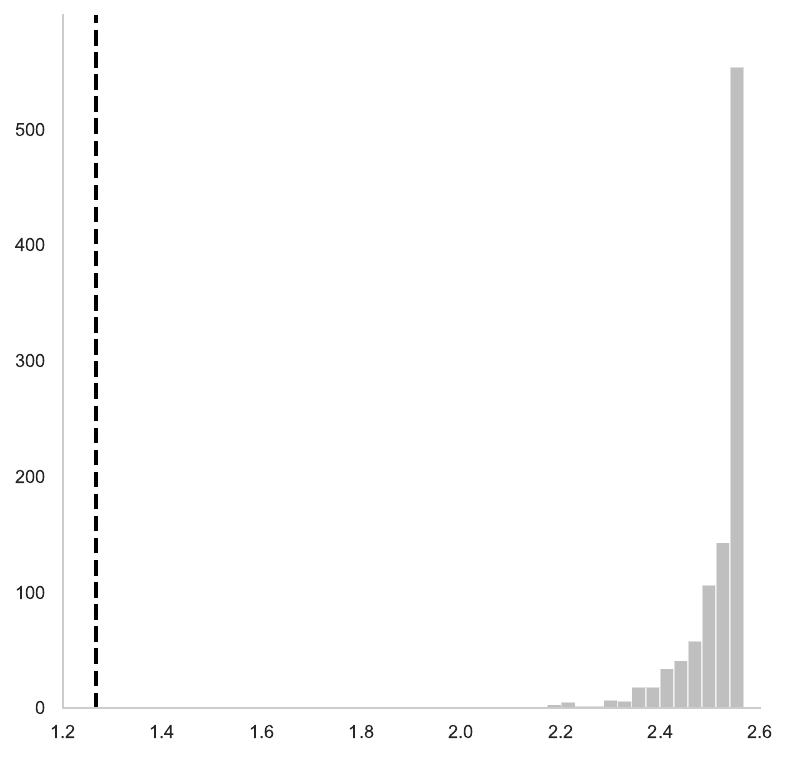}};

        \node at (2.95, -1) {$\mu = \mu_0$};
        \node at (2.95, -3) {\includegraphics[width=5.5cm]{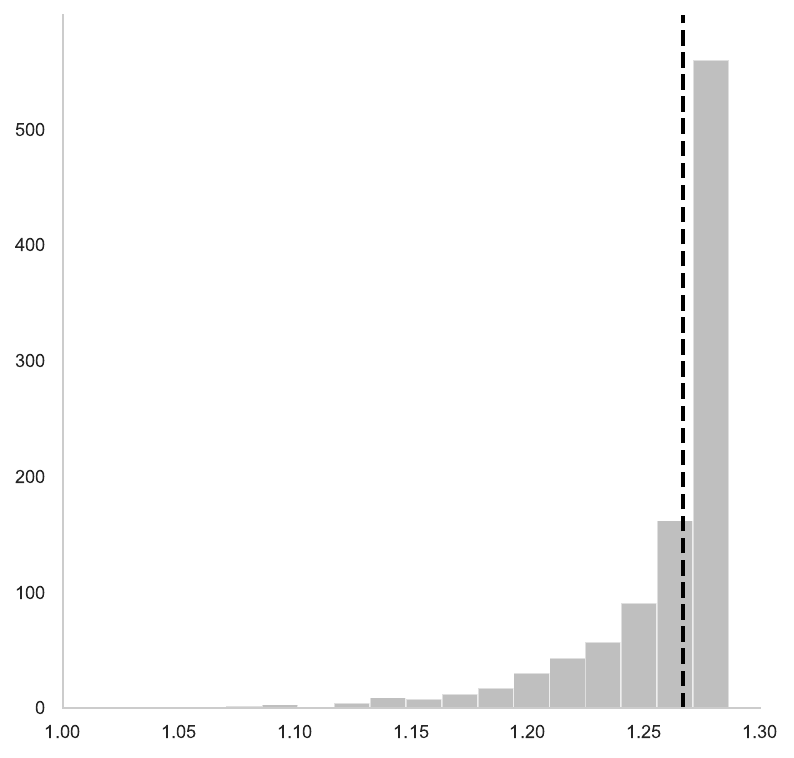}};

        \node at (6, -6) {\includegraphics[width=15cm]{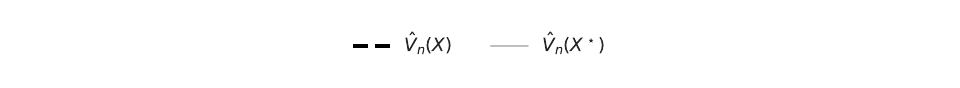}}; 
    \end{tikzpicture}
  
    \caption{For a fixed sample $X_1, \ldots, X_{50} \sim N(\mu_0, 1)$ with $\mu_0 = 1$, the panels display the distribution (grey histogram) of the randomized test statistic $\hat V_n(X^\star)$ with $1000$ randomizations for $\mu = \mu_0$ (left) and $\mu = 0 \neq \mu_0$ (right) against the value of the variance on the original data $\hat V_n(X)$ (dashed horizontal line).}\label{fig:real-variance}
  \end{figure}

The construction above provides a principled approach to testing symmetry on the real line using randomization techniques, without requiring parametric assumptions beyond symmetry. The choice of empirical variance as the test statistic is particularly natural: under symmetry around $\mu$, we expect the spread of observations to be balanced on either side of $\mu$. When the true mean differs from the hypothesized value $\mu$, the empirical variance tends to increase as observations are shifted away from $\mu$, making it sensitive to departures from the null hypothesis. Hence, the empirical variance is too extreme if it is too small compared to the randomized sample as illustrated in Figure \ref{fig:real-variance}. That is, for a sorted sample of randomized variances $V_{(1)}, \ldots, V_{(B)}$, the test is rejected at level $\alpha$ if $\hat V_n(X_1, \ldots, X_n) \leq V_{(k)}$ with $k = \lfloor \alpha B \rfloor$.

\subsection{Extension to Riemannian manifolds} \label{subsec-riemmanian}

The extension of these ideas to Riemannian manifolds presents several challenges. First, the notion of reflection needs to be appropriately generalized to account for the manifold's geometry. Second, the test statistic must be adapted to capture meaningful deviations from the null hypothesis while respecting the manifold structure. In particular, we will need to carefully consider how to define a variance-like quantity that preserves the desirable properties of sensitivity to asymmetry and invariance under the appropriate generalization of reflection. This section develops these extensions, showing how the fundamental principles of randomization testing can be preserved in the more general setting.

As a first generalization step, we will see how the intuition of the randomization procedure on the real line extends naturally to directional data on the unit circle $S^1 = \eset{ x \in \R^2 : \norm{x} = 1 }$. The distance between two points $x, y \in S^1$ is given by $d(x, y) = \arccos(x^\top y)$. Using $\angle(x, y)$ to denote the angle between the vectors $x$ and $y$, we have $x^\top y = \norm{x}\norm{y}\cos(\angle(x, y)) = \cos(\angle(x, y))$. Thus, the distance between two points on the circle corresponds to the angle between them. For a point $x \in S^1$, denote by $\theta_x \in [0, 2\pi)$ its angular representation with $\theta_{(1, 0)} = 0$. Directly generalizing the randomization procedure from the real line, define the reflection map in angular representation with
\begin{equation}\label{eq-reflection-circle}
  g_{\mu} \cdot \theta_x = 2\theta_\mu - \theta_x \mod 2\pi
\end{equation}
Without loss of generality, we can assume that $\mu = (1, 0)$, since this can be achieved without affecting distances via a rotation. The reflection map $g_\mu$ then corresponds to a sign flip of the angle or equivalently a the reflection of $y$ through the $x$-axis in the vector representation. This transformation does not change distances between points and is thus an isometry.

To generalize this construct to a complete connected $d-$dimensional Riemannian manifold $M$, we need to introduce the notions of geodesics, exponential and logarithm maps, for a thorough and rigorous introduction to Riemannian geometry, we refer the reader to \cite{chavel_cambridge_2006}. We will illustrate each of these concepts using the circle as a concrete example. Geodesics are curves $\gamma : I \to M$ defined over an interval $I$ that are locally length-minimizing, playing the role that straight lines do in Euclidean space. At any point $p \in M$, the tangent space $T_pM$ is a $d-$dimensional vector space that can be thought of as containing all possible velocities of curves passing through $p$. For the circle, the tangent space at any point is simply a line -- isomorphic to $\mathbb{R}$ -- tangent to the circle at that point. The exponential map $\exp_p: T_pM \to M$ takes a tangent vector $v \in T_p M$ and follows the geodesic starting at $p$ with initial velocity $v$ for one unit of time, i.e., $\exp_p(v) = \gamma(1)$ where $\gamma$ is the unique geodesic such that $\gamma(0) = p$ and $\gamma'(0) = v$. On the circle, $\exp_p(v)$ corresponds to starting at $p$ and moving counterclockwise by an angle of $|v|$ radians if $v$ is positive, or clockwise if $v$ is negative. We call the \textit{cut locus} of $p$ the boundary of the set $\eset{ v \in T_p M \mid d(\exp_p v, p) = \norm{v} }$ and denote it by $\t{cut}(p)$. The distance of a point $q \in M$ to the cut locus is called the \textit{injectivity radius} of $p$ and is denoted by $\t{inj}(p)$. On the circle, the cut locus of any point $p$ is its antipodal point, $\t{cut}(p) = \eset{-p}$, and hence the injectivity radius at any point is $\t{inj}(p) = \pi$. Within the injectivity radius of $p$, the exponential map is a diffeomorphism between $T_p M$ and a neighborhood of $p$ in $M$ with inverse $\log_p$. That is, for any point $q \in M$ with $d(p, q) < \t{inj}(p)$, it holds that $\exp_p (\log_p q) = q$.

Given a point $\mu \in S^1$, we can now define the map $g_\mu$ in terms of the exponential and logarithm maps. For $x \in S^1$, the reflection $g_\mu \cdot x$ provided in \eqref{eq-reflection-circle} as
\begin{equation} \label{eq-reflection-circle-exp}
  g_\mu \cdot x =
  \begin{cases}
    x &\t{if } x \in \t{cut}(\mu), \\
    \exp_\mu(-\log_\mu(x)) &\t{otherwise}.
  \end{cases}
\end{equation}
This map is also referred to in the literature as the \textit{geodesic symmetry} since $g_\mu \cdot x = \gamma_{-}(1)$ where $\gamma_{-}$ is the geodesic with $\gamma_{-}(0) = \mu$ and $\gamma_{-}'(0) = -\log_\mu x$. The randomization procedure in \eqref{eq-r-randomization} can be written on the circle more generically as first sampling a transformation $\mathbf{g}$ from the set $G_\mu = \eset{ \t{id}, g_\mu }$ and setting $X^\star = \mathbf{g} \cdot X$. The set of mappings $G_\mu$ corresponds to the set of all isometries preserving 0 in the tangent space of $S^1$ in $\mu$.

For tangent space of dimension $d > 1$, alternative transformations of the vector space can be possible. Consider a $d-$dimensional Riemannian manifold $M$, then for any point $p \in M$, the tangent space $T_p M$ is isomorphic to $\R^d$ and we can consider any isometry on $\R^d$ mapping the 0 vector to itself as candidates for the randomization maps: this is the orthonormal group $O(d)$ containing the rotations and reflections. This allows to define a set of mapping on $M$ preserving the test mean $\mu \in M$ via $G_\mu = \eset{ g_\mu^Q : Q \in O(d) }$ where
\begin{equation*}
  g_\mu^Q \cdot x =
  \begin{cases}
    x &\t{if } x \in \t{cut}(\mu), \\
    \exp_\mu(Q \log_\mu(x)) &\t{otherwise}.
  \end{cases}
\end{equation*}

In general a Riemannian manifold, even with further standard regularity conditions, this set of maps is not as well-behaved as on the real line or on the circle. The maps $g_\mu^Q$ still map $\mu$ to itself but are not always isometries themselves -- even if $Q$ is. Furthermore, even if the random variable $X$ is symmetric around its Fréchet mean $\mu$ and is almost surely within the injectivity radius of $\mu$, the maps $g_\mu^Q$ do not necessarily preserve the Fréchet mean. However, if the Fréchet function of $X$ is convex, an adjacent concept of mean preservation still holds. Assuming that $X$ is almost surely contained within the injectivity radius of $\mu$ and that the Fréchet function $F_X$ defined in Equation \eqref{eq-frechet-fn-def} is convex, the gradient of $F_X$ in a point $p \in \t{inj}(\mu)$ is given by $\t{grad} F_X(p) = -2\expec{\log_p(X)}$ and the Fréchet mean of $X$ exists and solves the score equation $\expec{\log_\mu(X)} = 0$, see \cite{kendall_probability_1990, le_consistency_1998, eltzner_smeary_2018}. In this situation, any $g_\mu^Q \in G_\mu$  preserves the Fréchet mean $\mu$ by convexity together with
\begin{align*}
  \expec{\log_\mu( g_\mu^Q  \cdot X)} = \expec{ \log_\mu(\exp_\mu(Q \log_\mu(X))) } = Q \expec{\log_\mu(X)} = 0.
\end{align*}
While this support restriction may appear to be rather restrictive, we will show in the next generalization that for a large class of Riemannian manifolds, the maps $g_\mu^Q$ can be used to construct a valid randomization without having to assume that $X$ is almost surely within the injectivity radius of $\mu$.

While the previous discussion focused on manifolds where the cut locus plays a crucial role in defining our randomization maps, there exist important manifolds in statistical applications where the cut locus is empty. This simplifies the construction of the randomization procedure significantly, as illustrated by the following example of the space of symmetric positive definite (SPD) matrices equiped with the Bures-Wasserstein distance.

\begin{example}[Bures-Wasserstein]\label{ex-bures}
  Let $\mathcal{D}_2(\R^p)$ be the set of probability measures $\mu$ on $\R^p$ such that $\mu$ has a density and $\expec{X^2} < \infty$ for $X \sim \mu$. The Wasserstein distance \cite{panaretos_invitation_2020} of order 2 between two measures $\mu, \nu \in \mathcal{D}_2(\R^p)$ is defined as
  \begin{equation*}
    d_{W_2}^2(\mu, \nu) = \inf_{\gamma \in \Pi(\mu, \nu)} \expect{(X,Y) \sim \gamma}{\norm{X - Y}_2^2}
  \end{equation*}
  where $\Pi(\mu, \nu)$ is the set of all joint distributions on $\R^p \times \R^p$ with marginals $\mu$ and $\nu$. The Wasserstein distance can be used to define a distance $d_\mathcal{B}$ on the set of covariance matrices $\mathbb{S}^p_+$ via centered Gaussian distributions, resulting in the \textit{Bures-Wasserstein distance} $d_{\mathcal{B}}(A,B) = d_{W_2}(\mathcal{N}(0, A), \mathcal{N}(0, B))$. The metric space $(\mathbb{S}^p_+, d_{\mathcal{B}})$ is a Riemannian manifold called the \textit{Bures-Wasserstein space}, and the distance function has the following closed form expression
  \begin{equation*}
    d_\mathcal{B}(A, B)^2 = \t{tr} A + \t{tr} B - 2\t{tr} (A^{1/2} B A^{1/2})^{1/2}.
  \end{equation*}
  In this space, the cut loci are empty sets and the exponential map is a non-isomorphic diffeomorphism over the whole space. See \cite{bhatia_bureswasserstein_2019, chewi_gradient_2020} for more details on Bures-Wasserstein spaces.
\end{example}

The generalization presented in this section demonstrates how the basic reflection principle from $\R$ extends naturally to Riemannian manifolds through three key steps: (1) first extending to the circle $S^1$ where reflections correspond to angle reversals, (2) generalizing to arbitrary Riemannian manifolds using the exponential and logarithm maps to define geodesic symmetries, and (3) incorporating the full orthogonal group $O(d)$ as a richer example of transformations to capture all possible isometries of $\R^d$. While additional technical conditions are needed compared to the Euclidean case --- particularly regarding the cut locus and convexity of the Fréchet function --- the fundamental principle of constructing randomization maps that preserve both the test mean and the metric structure remains the same.
\section{Isotropic randomization in metric spaces}\label{sec-meantest}

The randomization tests presented earlier rely heavily on the existence of a vector space structure to express the randomization, either local for Riemannian manifolds or global for Euclidean spaces. However, many interesting metric spaces, such as graphs, trees, or stratified spaces, lack these differential geometric tools. In this section, we develop a general framework for randomization tests based on isometries that preserve the key properties of the Euclidean approach while being applicable to any metric space with sufficient symmetries. This generalization allows us to extend mean testing procedures to a broader class of spaces while maintaining theoretical guarantees on test validity and power.

\subsection{Random isotropies}

Let $(\Omega, d)$ be a metric space and denote by $\t{Iso}(\Omega)$ the set of all bijective isometries from $\Omega$ onto itself. This set is a group under composition, called the \textit{isometry group} of $\Omega$, and acts on $\Omega$ by $g \cdot x = g(x)$. For any $\mu \in \Omega$, the \textit{isotropy group} of $\mu$ is the subgroup $G_\mu$ of $\t{Iso}(\Omega)$ of isometries mapping $\mu$ to itself, $G_\mu = \eset{ g \in \t{Iso}(\Omega): g \cdot \mu = \mu }$. A relevant property of the Fréchet mean is that it is invariant under the action of an isometry and as the next proposition shows, the isotropy group of the Fréchet mean of a random variable $X$ preserves all the distance moments of $X$.

\begin{proposition}\label{thm-moments-pres}
  Let $(\Omega, d)$ be a metric space, $p \geq 2$ and $X \in L^p(\Omega)$ be a random variable in $\Omega$ with Fréchet mean $\mu$. Then, for any random isometry $\mathbf{g}$ independent of $X$ with support equal to a subgroup of $G_\mu$, the variable $\mathbf{g} \cdot X$ has the same Fréchet mean and distance moments as $X$, that is
  \begin{equation*}
    \expec{\mathbf{g} \cdot X} = \mu \qquad\t{and}\qquad \expec{d(\mathbf{g} \cdot X, \mu)^k} = \expec{d(X, \mu)^k},
  \end{equation*}
  for all $k \leq p$.
\end{proposition}
\begin{proof}
  We will first prove that $\mu$ is the minimizer of the Fr\'echet function of $\mathbf{g} \cdot X$. Let $\omega \in \Omega$, we have using that $\mathbf{g}$ is an isometry that
  \begin{align*}
    \expec{d(\mathbf{g} \cdot X, \omega)^2}
    &= \expec{d((\mathbf{g}^{-1} \circ \mathbf{g}) \cdot X, \mathbf{g}^{-1} \cdot \omega)^2}\\
    &= \expec{d(X, \mathbf{g}^{-1} \cdot \omega)^2}\\
    &= \expect{\mathbf{g}}{ \expect{X}{ d(X, \mathbf{g}^{-1} \cdot \omega)^2 } }
  \end{align*}
  Since $\mu$ minimizes $\omega \mapsto \expec{d(X, \omega)^2}$, we have that the inner expectation can be lower bounded by $\expec{d(X, \mu)^2}$. Using this bound and reversing the above computations we get
  \begin{align*}
    \expec{d(\mathbf{g} \cdot X, \omega)^2}
    &\geq \expec{d(X, \mu)^2} \\
    &= \expec{d(\mathbf{g} \cdot X, \mathbf{g} \cdot \mu)^2} \quad&\text{(by isometry)}\\
    &= \expec{d(\mathbf{g} \cdot X, \mu)^2}&\text{(since } \mathbf{g} \in G_\mu)
  \end{align*}
  Therefore, $\expec{d(\mathbf{g} \cdot X, \mu)^2} \leq \expec{d(\mathbf{g} \cdot X, \omega)^2}$ for all $\omega \in \Omega$ proving the mean invariance $\expec{\mathbf{g} \cdot X} = \mu$. The claim concerning the distance moments is also a direct consequence of the fact that $\mathbf{g} \in G_\mu$,
  \begin{equation*}
    \expec{d(\mathbf{g}\cdot X, \mu)^k} = \expec{d((\mathbf{g}^{-1} \circ \mathbf{g})\cdot X, \mathbf{g}^{-1}\cdot \mu)^k} = \expec{d(X, \mu)^k}.
  \end{equation*}
\end{proof}
Proposition \ref{thm-moments-pres} suggests that the isotropy group of the Fréchet mean of $X$ is a natural candidate for the set maps to use for the randomization procedure. Given a random variable $\mathbf{g} \sim P_\mathbf{g}$ over $G_\mu$, possibly with support equal to a subgroup of $G_\mu$, we define the \textit{isotropic randomized variable} $X^\star = \mathbf{g}\cdot X$.

An example of metric spaces with an isotropy group allowing for a randomization procedure similar to the one presented previously are the globally symmetric spaces, a certain kind of Riemmanian manifolds in which the reflection map acts as an isomorphism.

\begin{example}[Globally symmetric spaces]
  Let $M$ be a Riemannian manifold; $M$ is called \textit{Riemannian globally symmetric} if each $p \in M$ is an isolated fixed point of a self-inverse isometry $s_p$, see \cite[Chapter IV]{helgason_differential_2001}. Hence the definition of a symmetric space directly corresponds to assuming that for any $\mu \in M$, the set $R_\mu = \eset{\t{id}, s_\mu}$ is a subgroup of the isotropy group $G_\mu$. Examples of globally symmetric spaces include the sphere $S^d$ and the hyperbolic space $\mathbb{H}^d$. In these spaces, a randomization procedure can thus be constructed by sampling $\mathbf{g}$ uniformly from $R_\mu$.
\end{example}

\subsection{Testing via isotropic randomization}

Proposition \ref{thm-moments-pres} ensures that under the null hypothesis $H_0: \expec{X} = \mu$, the randomized random variable $X^\star$ has the same Fréchet mean and distance moments as $X$. Under the alternative hypothesis $H_1: \expec{X} \neq \mu$, we can expect that transforming $X$ via the isotropy group of $\mu$ will not preserve the Fréchet mean and distance moments of $X$. This, together with the following theorem, suggests that the empirical Fréchet variance of $X^\star$ can be used as a test statistic for the hypothesis \eqref{eq-hypothesis}.

\begin{theorem}[Theorem 1 in \cite{dubey_frechet_2019}]\label{thm-clt-var}
  Suppose that the metric space $(\Omega, d)$ is bounded and satisfies Assumption \ref{ass-cover}. If additionally the random variable $X \in \Omega$ satisfies Assuption \ref{ass-separated}, then
  \begin{equation}
    \sqrt{n}(\hat V_n - \Var{X}) \rightarrow \mathcal{N}(0, \sigma^2_d)
  \end{equation}
  where $\sigma^2_d = \Var{d(X, \expec{X})^2}$
\end{theorem}

Since our goal is to construct a test based on the empirical variance of $X^\star$, we also need to make sure that it will have power against alternatives to the null hypothesis.

Let us consider two simple cases to illustrate where the variance alone might not be sufficient. First, let $\Omega$ be the Euclidean space $\R^2$; here, for any $x \in \R^2$ the isotropy group $G_x$ is given by the rotations and reflections, up to a conjugation with the translation to $x$. For any direction $x \in \R^2$, the reflection $r_x$ with respect to the line spanned by $x$ is $r_\mu(x) = S_\mu x$ with $S_\mu = 2\frac{\mu\mu^\top}{\mu^\top\mu} - \t{id}_2$. As a reflection $r_\mu$ is a self-inverse isotropy of $x$, and hence the set $R_x = \eset{ \t{id}, r_x }$ is a subgroup of $G_x$. For ease of computation and implementation we might consider restricting the support of the randomization group action to $R_\mu$ where $\mu$ is the tested mean. However now $\mu \in \R^2$ and $X \sim \mathcal{N}(0, \t{id}_2)$ be the standard normal distribution in $\R^2$. By elementary properties of the multivariate normal distribution, $\Var{g \cdot (X + \lambda \mu)} = \Var{g \cdot (X + \mu)}$ for all $\lambda \in \R$ and $g \in R_\mu$. Hence, the variance alone does not have power against colinear alternatives when restricting the support of $\mathbf{g}$ too much.

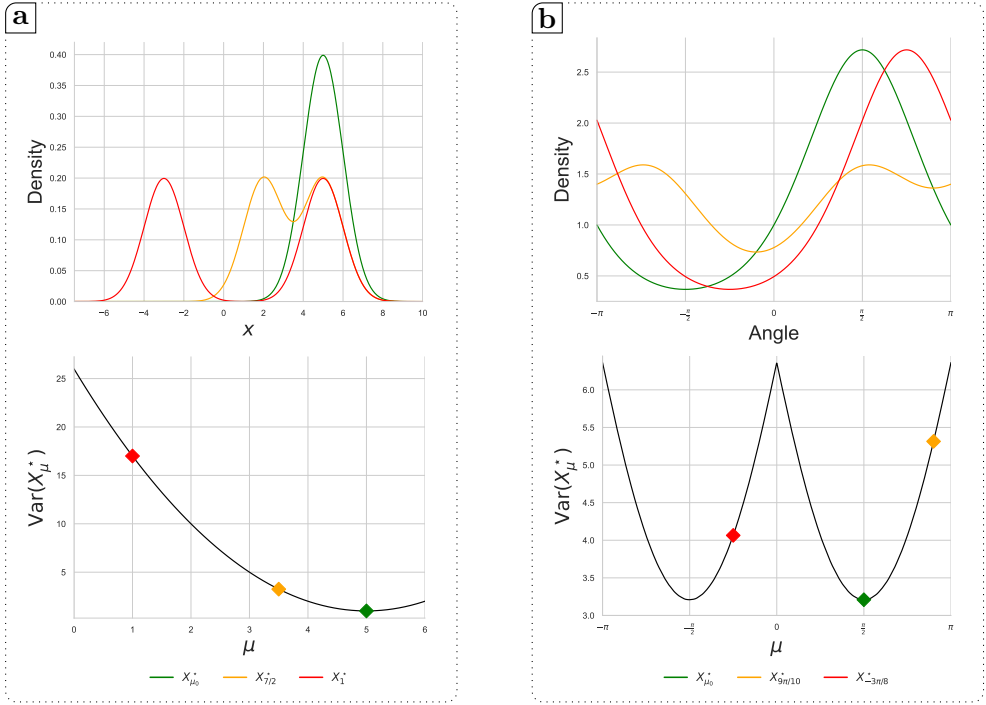
\begin{figure}[!t]
  \hspace*{-.8cm}
  \centering
  \begin{tikzpicture}
    \begin{scope}[shift={(-3.5, 0)}]
      \draw[dotted, rounded corners=4pt] (0, 0) -- ++(6, 0) -- ++(0,-9.3) -- ++(-6, 0) -- cycle;
      \node (Aa) at (2.95, -2.5) {\includegraphics[width=5.5cm]{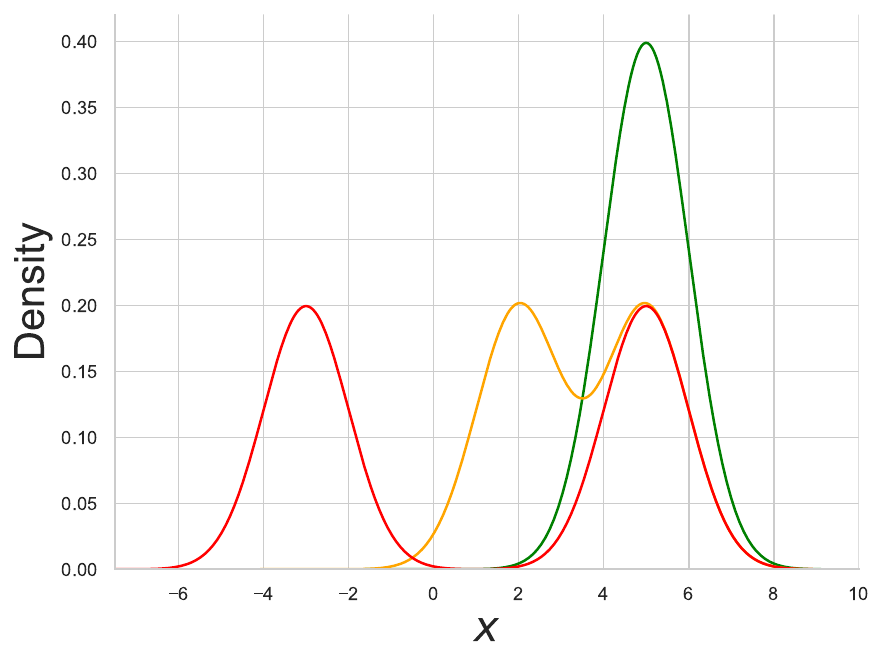}};
      \node (Ab) at (2.95, -6.7) {\includegraphics[width=5.5cm]{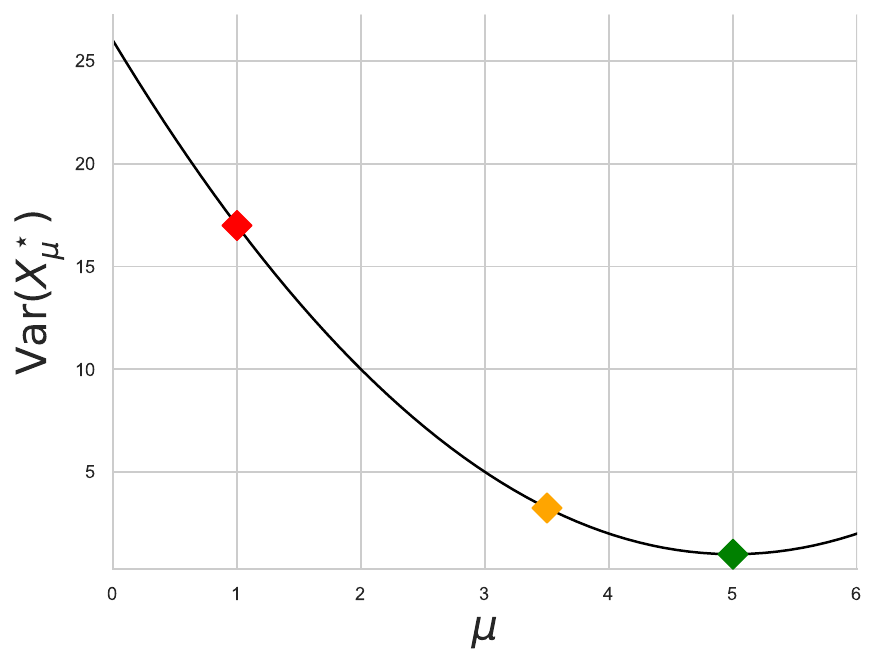}};
      \node (Ac) at (3.25, -9) {\includegraphics[width=8cm]{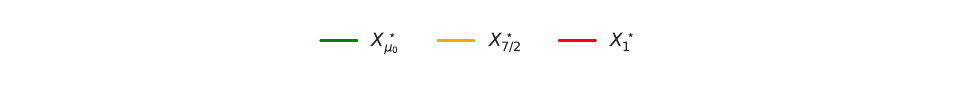}};
      \node[fill=white] (aa) at (0.2, -0.2) {\small\textbf{a}};
      \draw (0.4, 0) -- ++(0,-.4) -- ++(-.4,0) {[rounded corners=4pt]-- ++(0,.4)} -- cycle;
    \end{scope}

    \begin{scope}[shift={(3.5, 0)}]
      \draw[dotted, rounded corners=4pt] (0, 0) -- ++(6, 0) -- ++(0,-9.3) -- ++(-6, 0) -- cycle;
      \node (Ba) at (2.95, -2.5) {\includegraphics[width=5.5cm]{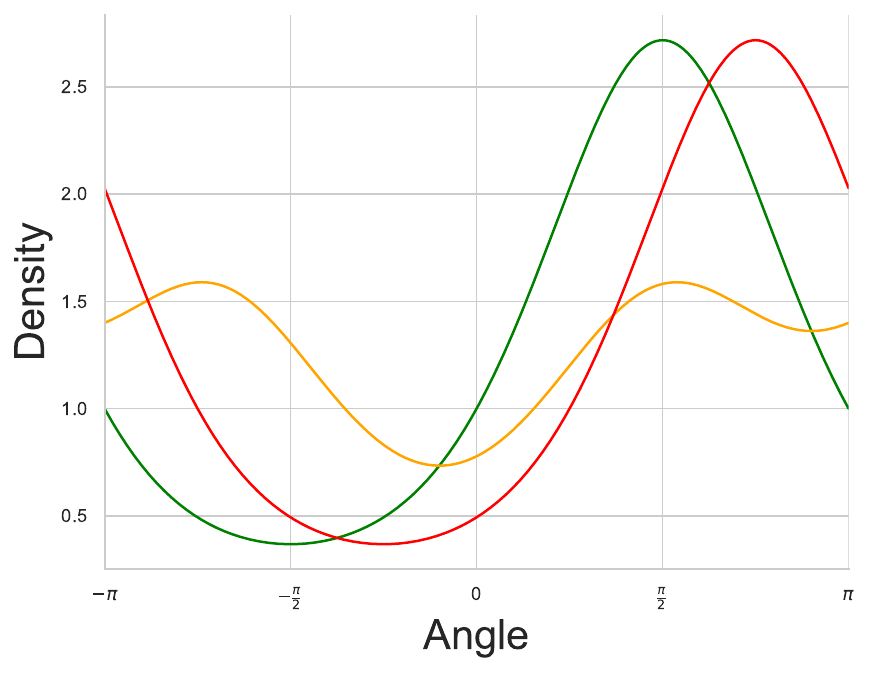}};
      \node (Bb) at (2.95, -6.7) {\includegraphics[width=5.5cm]{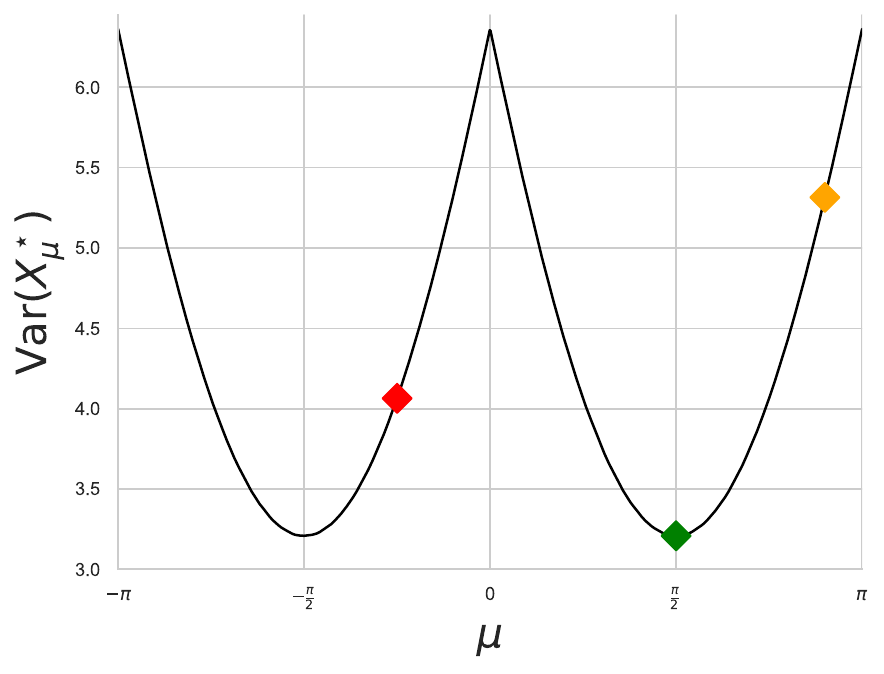}};
      \node (Bc) at (3.25, -9) {\includegraphics[width=8cm]{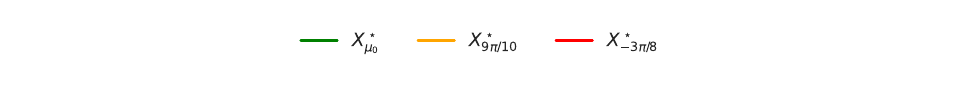}};
      \node[fill=white] (bb) at (0.2, -0.2) {\small\textbf{b}};
      \draw (0.4, 0) -- ++(0,-.4) -- ++(-.4,0) {[rounded corners=4pt]-- ++(0,.4)} -- cycle;
    \end{scope}
  \end{tikzpicture}

  \caption{Illustration of the behavior of a random variable and its variance under the reflective randomization proposed in \eqref{eq-reflection-circle-exp}. In Panel (a), the random variable follows a normal distribution $N(5,1)$ and is randomized around its true mean (green) as well as other hypothesized points $\mu$ on the real line (orange, red). Panel (b) considers a random variable on the circle $S^1$ following  a von Mises distribution $VM(\pi/2, 0.3)$. As in panel (a), we show the result of the randomization with respect to the true mean (green) as well as randomizations with respect to other points $\mu$ on the circle (orange, red). In both panels, the top plot displays the probability density functions of each variable, while the bottom plot shows their Fréchet variance as a function of the randomization point $\mu$. The colored points on the variance curves correspond to the specific distributions displayed above.}\label{fig:randomization}

\end{figure}

While the previous example might seem artificial, there exist spaces in which the entire isotropy group in a point does not identify the point itself. For instance on the example of the circle $S^1$ considered above, the isotropy group in any point $x \in S^1$ is $G_x = \eset{ \t{id}, g_x }$, where $g_\mu$ is the angular reflection given in Equation \eqref{eq-reflection-circle}. There, we have that the antipodal point $x^-$ of $x$ has the same isotropy group, $G_{x^-} = G_x$. Thus, whether the randomization is done to test the null $H_0 : \expec{X} = x$ or $H_0 : \expec{X} = x^-$, the testing procedure will randomize the data in the exact same way making the two situations indistiguishable. This behavior on the circle is illustrated in Figure \ref{fig:randomization} where we compare the behavior of the randomization in the simple case of the real line to the circle in which points are not fully identified by their isotropy group. The same argument can be extended to higher dimensional spheres where using the variance of the randomized distribution will not have power against antipodal alternatives. 

These two examples motivate the notion of an \textit{admissible randomization} for constructing a test based only on the empirical Fréchet variance. In simple words, an admissible randomization should make it possible to distinguish between the null hypothesis and any other point in the space.

\begin{definition}
  Let $(\Omega, d)$ be a metric space and $\mu \in \Omega$. A random group action $\mathbf{g}$ over $\Omega$ is called $\mu-$admissible randomization if the support of $\mathbf{g}$ is a subgroup of isotropy group $G_\mu$ and there exists no points $x \in \Omega \backslash \eset{\mu}$ such that $\mathbf{g} \cdot x$ is almost surely constant.
\end{definition}

With this definition, the following proposition shows that the admissible randomizations result in a test with power against alternatives.

\begin{algorithm}[t]
  \caption{Isotropic Randomization Test for Fréchet Mean}\label{alg:iso-rand}
  \begin{algorithmic}[1]
  \Require Sample $X_1,\ldots,X_n$, hypothesized mean $\mu$, randomization distribution $P_{\mathbf{g}}$, number of replicates $B$, significance level $\alpha$
  \For{$b = 1,\ldots,B$}
      \State Sample independently $\mathbf{g}_1,\ldots,\mathbf{g}_n \simiid P_{\mathbf{g}}$
      \State Construct randomized sample $X^\star_i = \mathbf{g}_i \cdot X_i$ for $i=1,\ldots,n$
      \State Compute test statistic $V_b = \hat{V}_n(X^\star_1,\ldots,X^\star_n)$
  \EndFor
  \State Sort the test statistics: $V_{(1)} \leq \cdots \leq V_{(B)}$
  \State \Return Reject $H_0$ if $\hat{V}_n(X_1,\ldots,X_n) > V_{(\floor{\alpha B})}$
  \end{algorithmic}
  \end{algorithm}

\begin{proposition}
  Let $X \in L^2(\Omega)$ be a random variable satisfying Assumption \ref{ass-separated}. For any $\mu \in \Omega$ and $\mu-$admissible randomization $\mathbf{g}$, let $X^\star = \mathbf{g} \cdot X$. Then, $\Var{X^\star} \geq \Var{X}$ with equality if and only if $\expec{X} = \mu$.
\end{proposition}
\begin{proof}
  Let $\mu^\star = \expec{X^\star}$ and $\mu_X = \expec{X}$. By definition of $X^\star$ and the admissibility of $\mathbf{g}$, we have that
  \begin{align*}
    \Var{X^\star}
    &= \expec{d(\mathbf{g} \cdot X, \mu^\star)^2}
    = \expec{d((\mathbf{g}^{-1} \circ \mathbf{g}) \cdot X, \mathbf{g}^{-1} \cdot \mu^\star)^2}\\
    &= \expec{d(X, \mathbf{g}^{-1} \cdot \mu^\star)^2}
    = \expect{\mathbf{g}}{ \expect{X}{ d(X, \mathbf{g}^{-1} \cdot \mu^\star)^2 } }
  \end{align*}
  By definition of the Fréchet mean, we have that $\expec{d(X, \omega)^2} \geq \expec{d(X, \mu_X)^2}$ for all $\omega \in \Omega$ and hence
  \begin{equation*}
    \Var{X^\star} \geq \expec{d(X, \mu_X)^2} = \Var{X},
  \end{equation*}
  Since $\expec{d(X, \mu_X)^2}$ lower bounds the Fréchet function and by Assumption \ref{ass-separated} the Fréchet mean is well separated, equality holds if and only if $\mathbf{g}^{-1} \cdot \mu^\star = \mu_X$, or equivalently $\mu^\star = \mathbf{g} \cdot \mu_X$ holds almost surely. Since $\mathbf{g}$ is $\mu-$admissible, this can only hold if $\mu_X = \mu$.
\end{proof}

Hence, if an admissible randomization exists, the empirical variance $\hat V_n$ defined in \eqref{eq-var-emp} can be used as a test statistic for the hypothesis \eqref{eq-hypothesis}. We formalize this testing procedure in Algorithm \ref{alg:iso-rand}.

The practical implementation of Algorithm \ref{alg:iso-rand} requires constructing isometries that form an admissible randomization. While these arise naturally in globally symmetric spaces, more complex spaces that lack global symmetry require careful consideration of the local geometry. We now examine an example of such a space --- a stratified space formed by gluing together Euclidean halfspaces, which naturally arises in applications involving branching structures or networks.

\begin{example}[Booklets]\label{ex-booklet}
  Let $k \in \N$, we define the $d$-dimensional booklet $B_d^k$ as $k$ copies of the halfspace $\R_+ \times \R^{d-1}$, glued together along $\eset{0} \times \R^{d-1}$. Each point is then represented by a triplet $(z, x, y)$ where $z \in \eset{1, \ldots, k}, x \in \R_+$ and $y \in \R^{d-1}$. The distance between two points $(z, x, y)$ and $(z', x', y')$ is given by
  \begin{equation*}
    d((z, x, y), (z', x', y'))^2 =
    \begin{cases}
      (x - x')^2 + \norm{y - y'}_2^2 & \text{if $z = z'$}, \\
      (x + x')^2 + \norm{y - y'}_2^2 & \text{otherwise}.
    \end{cases}
  \end{equation*}
  For a point $\mu = (z, x, y)$, one can construct a reflection map with two ingredients: an $y-$isotropy $g_y \in G_y$; a self-inverse permutation $\pi_z$ on $\eset{1, \ldots, k}$ such that $\pi_z$ maps $z$ to itself, $\pi_z z = z$. The reflection map $g_\mu$ is then defined as
  \begin{equation*}
    g_\mu(z', x',y') = (\pi_z z', x', g_y y').
  \end{equation*}
  An illustration of the booklet $B_2^4$ is shown in Figure \ref{fig:booklet}.
\end{example}

While no proof of the consistency of the test is available here, it is worth noting that the randomization procedure provides a form of consistency under the null hypothesis. Indeed, if $\mu$ is the Fréchet mean of $X$, and under the assumptions presented above, the empirical Fréchet mean is consistent with $d(\mu, \hat\mu_n) = o_p(1)$, then for any $g \in G_\mu$
\begin{equation*}
  d(\hat\mu_n, g \cdot \hat\mu_n) \leq d(\hat\mu_n, \mu) + d(\mu, g \cdot \hat\mu_n) = d(\hat\mu_n, \mu) + d(g^{-1} \cdot \mu, \hat\mu_n) = 2 d(\hat\mu_n, \mu).
\end{equation*}
Hence, a $n \rightarrow \infty$, the isotropy group $G_\mu$ almost fixes the empirical Fréchet mean, and the randomization procedure will almost fix the empirical Fréchet variance. This suggests that the variability introduced by the randomization disappears as the sample size grows under the null, and the randomization procedure will not introduce variability, leading to not rejecting the test.

The framework developed above provides a general approach for testing hypotheses about Fréchet means in metric spaces. A natural question is whether we can identify classes of distributions where this approach is particularly well-suited. Radially symmetric distributions, which we examine next, form such a class --- their inherent symmetry properties align naturally with the isometric randomization procedure, making them an ideal setting for applying these tests. Moreover, studying these distributions helps us better understand the relationship between geometric symmetry and statistical inference in metric spaces.

\subsection{Application to radially symmetric distributions}\label{subsect-radial}

Our framework is particularly well-suited for \textit{radially symmetric distributions} on metric spaces. We introduce a novel definition of radially symmetric distributions which generalizes the definition in Euclidean spaces. Recall that for $\mu \in \R^d$, a distribution $P_\mu$ with density $f$ with respect to the Lebesgue measure is called  radially symmetric around $\mu$ if its density can be written as $f(x) = h(d(x, \mu))$ for some $h : \R \rightarrow \R$. This definition naturally extends to Riemannian manifolds by assuming a radial density with respect to the Riemannian volume measure. A common example of such a distribution is the von Mises-Fisher distribution on the $d-$dimensional sphere $S^d$ with location and concentration parameters $\mu \in S^d$ and $\kappa > 0$. The density of this distribution with respect to the volume measure is $f_{\mu, \kappa}(x) \propto \exp(\kappa x^\top\mu) = \exp(\kappa d(x,\mu))$, which is invariant under isotropies of $\mu$. In homogeneous Riemannian manifolds, the Fréchet mean of a symmetric distribution is equal to its point of symmetry $\mu$, see \cite[Theorem 2]{mccormack_equivariant_2023}. In order to avoid having to work with the technicalities of generalizing the measure volume, we define the notion of radially symmetric distributions on metric spaces via the group action of the isometry group.

\begin{definition}
  A random variable $X$ on $\Omega$ is called \textit{radially symmetric} around $\mu \in \Omega$ if its distribution is invariant under the action of any $g \in G_\mu$. That is, $g \cdot X \stackrel{\mathcal{D}}{=} X$ for every $g \in G_\mu$.
\end{definition}

This definition of radially symmetric distributions is consistent with the previous one in Euclidean spaces. Indeed, consider a random variable $X \sim P_\mu$, radially symmetric around $\mu \in \R^d$, with density $f(x) = h(d(\mu, x))$ with respect to the Lebegues measure. For any $g \in G_\mu$ a change of variable argument shows that the random variable $g \cdot X$ has density $f_g(y) = h(d(g^{-1} \cdot y, \mu))$. By the fact that $g \in G_\mu$, this gives $f_g(y) = h(d(y, g(\mu))) = h(d(y, \mu))$ and hence $g \cdot X \stackrel{\mathcal{D}}{=} X$.

\begin{example}[Normal distribution on the spider]
  Consider the booklet with $d=1$ and $k \in \N$ branches defined in Example \ref{ex-booklet}, also called a \textit{spider}, constructed by gluing together $k \in \N$ copies of the real line through the origin. There, given a point $\mu = (z, x)$, a simple radially symmetric distribution can be constructed by considering mixture of normal distributions $P_\mu$ on the real line. Sample the branch index $Z$ uniformly from $\eset{1, \ldots, k}$ and sample $X \mid Z=z \sim N(x, 1)$ and $X \mid Z \neq z \sim N(0, 1)$. The isotropy group of $\mu$ is only composed of transformations that map the branch $z$ to itself and, when applied to the other branches, applies a permutation to the branch index and a reflection on the position. By construction, $P_\mu$ is invariant under all isotropies of $\mu$ and is thus symmetric in $\mu$.
\end{example}

It is not clear whether the Fréchet mean of a radially symmetric distribution is the point of symmetry as in the case of homogeneous Riemannian manifolds. This stems from the same reason as why some randomizations are not admissible: If two points share the same isotropy group, a random variable will either be radially symmetric around both points or neither. However, the following lemma shows that if the isotropy group of the point of symmetry of a symmetric distribution is not contained in the isotropy group of any other point of $\Omega$, the Fréchet mean of the distribution is the point of symmetry.

\begin{lemma}
  Let $X \in L^2(\Omega)$ be a radially symmetric distribution around $\mu \in \Omega$ with Fréchet mean $\mu_X$. If the isotropy group $G_\mu$ is not contained in the isotropy group of any other point in $\Omega$, then $\mu_X = \mu$.
\end{lemma}
\begin{proof}
  This is a direct consequence of the equivariance of the Fréchet mean described in \cite{mccormack_equivariant_2023}: for any $g \in \t{Iso}(\Omega)$, we have that $\expec{g \cdot X} = g \cdot \expec{X}$. Hence, by symmetry of $X$ around $\mu$ we have for all $g \in G_\mu$ that $g \cdot X \stackrel{\mathcal{D}}{=} X$ and hence $\expec{g \cdot X} = \mu_X$. Therefore by equivariance $\mu_X = \expec{g \cdot X} = g \cdot \mu_X$ implying that $g \in G_{\mu_X}$ and hence $G_\mu \subset G_{\mu_X}$. Since $G_\mu$ is not contained in the isotropy group of any other point, this implies that $\mu_X = \mu$.
\end{proof}

\section{Numerical experiments}\label{sec-numerical}

We now illustrate our theoretical results with numerical experiments taking place in metric spaces with different properties and challenges for the test. We consider the example of the circle, with a both a symmetric distribution and a non-symmetric mixture distribution. In the second example, we explore the behavior of the test on a bounded subspace of the $2-$dimensional booklet described in Example \ref{ex-booklet}. Finally, we consider the space of symmetric positive definite (SPD) matrices equipped with the Bures-Wasserstein distance.

In each scenario, we consider a fixed $\mu_0 \in \Omega$ to test for and generate datasets of size $n \in \eset{100, 200, 400}$ with true Fr\'echet mean $\mu$, where the values of $\mu$ are elements on the geodesic ray connecting $\mu_0$ and a chosen $\mu_1 \neq \mu_0$. This allows us to consider the performance of the test as a function of the distance between $\mu_0$ and $\mu$. Additionally, we evaluate the power of the test as a function of $n$ for $n \in \eset{100, 200, 400, 600}$. For each metric space, we generate $500$ datasets. For situation, we run the isotropic test at level $\alpha = 0.05$ with $B = 1000$ resampling replicates and record the rejection rate.

All simulations and analyses are done in Python. The code to reproduce the experiments and figures is available online\footnote{\url{https://github.com/matthieubulte/meantesting}}.

\subsection{Directional data}\label{subsec-num-circle}
\input{circle_fig}
\input{circle_mix_fig}

In the first experimental setup, we consider the space of directions on the circle $\Omega = S^1$. We generate data from two different distributions: a von Mises distribution $VM(0, 0.3)$ and a mixture $P = (1 - p)VM(\pi/2, 0.3) + p VM(0,0.3)$ with mixing proportion $p = 1/3$. In both cases, for a given tested Fréchet mean $\mu_0$, we construct the randomization based on the subgroup of $G_{\mu_0}$ consisting of the identity and the reflection map described in Equation \eqref{eq-reflection-circle} and sample the isotropy $\mathbf{g}$ uniformly over this subgroup. In terms of the randomized sample, this corresponds with equal probability to either the original data point or the data point reflected around $\mu_0$. We compare out test results to the score test \cite[Chapter 7]{mardia2009directional} which uses the approximate normality of the circular mean.

In the first case, $VM(0, 0.3)$ is radially symmetric around its mean $\pi/2$, as discussed in Section \ref{subsect-radial}. The alternativc hypotheses $H_1$ are expressed in terms of the angle difference $\delta$ with $\mu_\delta = \expec{X} + \delta = \delta$. For testing the power of the test against local alternative, we sample according to the von Mises distribution with mean $0.2  n^{-1/2}$ for different $n$ considered. The visualization of the data, sample and resulting size and power are shown in Figure \ref{fig:circle-vm}. Panel (b) shows that the test appears to be consistent: it has correct size $0.05$ under the null hypothesis ($\delta = 0$) and rapidly converges to 1 as $\abs{\delta}$ increases. While not proved in the previous sections, Panel (c) shows that the test has power against $\sqrt{n}-$alternatives in this scenario. Omitted here is the performance of the score test which is expected considering that the sampling distribution is a von Mises.

In the second case of the mixture of von Mises, the true Fr\'echet mean of is $\pi/3$ but the unequal weighting of the two mixture components renders the distribution non-symmetric around its mean. Here, alternative distributions are generated by shifting the mean of both components of the mixture, giving $P_\delta = (1 - p)VM(\pi/2 + \delta, 0.3) + p VM(\delta,0.3)$. The data space and distribution along with experiment results are displayed in Figure \ref{fig:circle-mix}. While the distribution is not symmetric, the experimental results show a similar behavior as in the symmetric case, and the test has both correct size under the null and power against local alternatives. While we do not illustrate it here, we found that the score test in this scenario appears to be strongly biased and fails to maintain size or power for any of the tested sample sizes. The difference to the previous scenario is likely due to the mixture distribution.

\subsection{Bures-Wasserstein distance on SPD matrices}\label{subsec-num-cov}
\input{bw_fig}

We now consider the space of $2\times 2$ SPD matrices $\Omega = \mathbb{S}^2_{+}$ equipped with the Bures-Wasserstein distance, described in Example \ref{ex-bures}. To construct a sample with a given a mean $\mu \in \mathbb{S}^2_{+}$ on the Bures-Wasserstein manifold, we first sample a random element in the tangent space $T_\mu$. The tangent space at a point $\mu \in \mathbb{S}^2_{+}$ is the set of symmetric matrices \cite{bhatia_riemannian_2006}, so we sample the element of the tangent space in $\mu$ by first sampling a $2\times 2$ matrix $M$ with i.i.d.\,standard normal entries and symmetrize it via $V = (M + M^\top) / 2$. The sample point is then constructed by applying the exponential map, giving $X = \exp_\mu V$. An explicit form of the exponential map can be written in terms of the Lyapunov operator \cite{malago_wasserstein_2018}. In the examples, we chose $\mu_0 = \mathbbm{1}_2$ giving closed forms for the exponential and logarithm maps, $\exp_{\mu_0} V = (V/2 + \mathbbm{1}_2)(V/2 + \mathbbm{1}_2)$ and $\log_{\mu_0} M = 2M^{1/2} - 2 \mathbbm{1}_2$. Similarly to the previous example, the randomization is done by sampling $\mathbf{g}$ uniformly over $\eset{\t{id}, g_{\mu_0}}$ where $g_{\mu_0}$ is the reflection defined in Equation \eqref{eq-reflection-circle-exp} in terms of the exponential and logarithm maps.

To assess the power and size of our test, we sample under alternative distributions generated by sampling from the same process, with a Fréchet mean $\mu_\delta$ lying on the geodesic ray passing through the identity matrix $\mu_0 = \mathbbm{1}_2$ and $\mu_1 = \big(\begin{smallmatrix}4 & 1\\ 1 & 3 \end{smallmatrix}\big)$, where $\delta \in [-1, 1]$. A similar display of the experiment as in the previous numerical experiments is found in Figure \ref{fig:bw}. The experiment shows that the isotropic test is able to detect the change in the mean of the distribution for $\abs{\delta} > 0$ and maintain level, corresponding to $\delta = 0$. While the power increases with sample size for a fixed alternative, shown in Panel (c), our numerical results suggest that the test lacks power against $\sqrt{n}$-local alternatives (not shown).

\subsection{Booklet}
\input{booklet_fig}

In the final experiment, we consider the booklet $B^4_2$,  described in Example \ref{ex-booklet}. As a reminder, the space $B^4_2$ is constructed by gluing together four \textit{branches} (copies of $\R_+$) via the origin $0$, a $4-$spider, and attaching to each point on this structure a copy of the real line $\R$. Hence each point in $B^4_2$ can be represented via three coordinates $(z, x, y)$: the index of the branch $z \in \eset{1, 2, 3, 4}$, the position on the branch $x \in \R_+$, and the position on the real line $y \in \R$. A visualization of this space can be found in Panel (a) of Figure \ref{fig:booklet}.

To sample data from $B^4_2$, we consider the following hierarchical model where the distribution of $X$ is determined by the branch $Z$ on which the point lies and $Y$ is independent of $X$ and $Z$,
\begin{equation}\label{eq-booklet-distrib}
  \begin{aligned}
    Z &\sim \t{Categorical}\left(\frac{4}{7}, \frac{1}{7}, \frac{1}{7}, \frac{1}{7}\right)\\
    Y \mid Z &\sim \mathcal{N}(1, 1)\\
    X \mid Y, Z &\sim
    \begin{cases}
      \t{Beta}(20, 5) \qquad &\t{if}\, Z = 1,\\
      \t{Beta}(5, 20) \qquad &\t{if}\, Z \in \eset{2, 3, 4}
    \end{cases}.
  \end{aligned}
\end{equation}
Minimizing the Fréchet function over these distributions yields that the Fréchet mean of this distribution has $\mu^z = 1$, $\mu^x = 0$ and $\mu^y = 1$. The random isotropy $\mathbf{g}$ is chosen as described in Example \ref{ex-booklet}. First, a random permutation $\mathbf{g_z}$ over $\eset{1, 2, 3, 4}$ is drawn uniformly over the set of randomizations with fixed point $\mu_z$. Given the geometry of the spider, there is no isotropy on the spider that modifies the $x$ component. The $y$ component is changed via $\mathbf{g_y}$ uniformly sampled over $\eset{\t{id}, g_{\mu_y}}$ where $g_{\mu_y}$ is the reflection on the real line defined in Section \ref{subsec-realline}. All in all, we get that for a random $X \in B^4_1$, its randomization is given by $\mathbf{g} \cdot X = (\mathbf{g_z}\cdot X_z, X_x, \mathbf{g_y}\cdot X_y)$.

In this scenario, there is no obvious way to change the sampling process in to obtain a desired mean. Thus, we keep the sampling process fixes and vary the null hypothesis considered instead. The test is run for $H_0: \mu = \mu_\delta$ where $\mu_\delta$ is on the geodesic between the true $\mu_0$ of the data generating process and $\mu_1 = (2, 1, 0)$. The results of the experiment are displayed in Figure \ref{fig:booklet}. Similarly to the previous experiments, we observe that the isotropic test is able to detect the change in the mean of the distribution for $\delta > 0$ and maintain level. Panel (c) suggests that the test here has power against $\sqrt{n}-$alternatives. 
\input{wind_fig.tex}
\section{Application: Wind in Western Denmark} \label{sec-realdata}

As an illustrative example, we analyze wind direction data from the Danish Meteorological Institute (DMI). We are interesting in testing whether the mean wind direction aligns with the documented south-westerly pattern in the region \cite{doi:10.1126/science.1169823}, $\mu_0 = 5\pi/4 = 225^\circ$. The data is obtained from the Blåvandshuk Fyr station, at the most western point of Denmark. The station records the wind direction hourly, with direction measured in degrees from North, represented by elements on the circle $S^1$. Each measurement consists of a single reading from a wind vane, with values ranging from 0 to 359 degrees, representing the direction from which the wind is blowing, meaning that a wind blowing from North to the South would correspond to a measurement of $\pi/2$. In order to reduce the dependence between the observations and to avoid systematic time-of-day trends, we only consider a subset of the data consisting of measurements taken at 12pm every day in from June to November 2024, resulting in 152 observations. 

The dataset is displayed in Panel (a) of Figure \ref{fig:wind}, where the wind directions appear to be evenly distributed around the empirical Fréchet mean $\hat\mu_n \approx 3.944 \approx 225.98^\circ$, very close to the hypothesized south-western direction. The proximity of the empirical Fréchet mean to the hypothesized mean $\mu_0$ provides initial qualitative evidence supporting the hypothesis. We further investigate whether the data is symmetrically distributed around $\mu_0$. Panel (b) shows kernel density estimates of the signed angles from $\mu_0$ to $X$ and $\mu_0$ to $g^{\mu_0} \cdot X$ (after reflection around $\mu_0$). The deviation between these densities indicates that $X$ is not invariant under isotropies of $\mu_0$, suggesting a deviation from radial symmetry around $\mu_0$.

We proceed to test the null hypothesis $H_0: \mu = \mu_0$ against the alternative $H_1: \mu \neq \mu_0$. We use the isotropic test specialized to the circle $S^1$ described in Section \ref{subsec-riemmanian} with $B = 1000$ randomizations. The result of the test is displayed in Panel (c) of Figure \ref{fig:wind}. The empirical CDF of the test statistic under randomization is displayed, along with the observed value of the statistic $T$ on the original sample. The observed value of the test statistic is $T \approx 1.743$, corresponding to an approximated p-value of $\hat p_n \approx 0.64$ which is not significant at the $5\%$ level. This result does not provide evidence against the null hypothesis, suggesting that the typical wind direction at Blåvandshuk Fyr is indeed following the south-westerly pattern. 

In comparison, we can consider different alternatives for testing this problem that one might consider. A first naive idea, based on running a standard $t$-test on the raw angles, estimates the mean angle at approximately $-25^\circ$ giving a statistic $t = -27$ and a $p$-value numerically equal to 0, thus rejecting the null hypothesis. This illustrates the importance of considering the circular nature of the data. A standard alternative in the circular data literature is the score test \cite[Chapter 7]{mardia2009directional} based on the asymptotic normality of the circular mean. This gives a $\chi^2_1$ statistic of approximately $0.003$ which corresponds to a $p$-value approximately equal to 0.046, rejecting the null hypothesis at level $\alpha = 0.05$. 

One can also consider a notion of confidence interval by considering the set of means that are not rejected by the test. We do this by finding an interval $\mathcal{C}_{(1-\alpha)}$ containing $\mu_0$ such that for all $\mu \in \mathcal{C}_{(1-\alpha)}$ the test does not reject the null hypothesis $H_0 : \expec{X} = \mu$ at level $\alpha$. We construct this interval by evaluating the test on a grid of 1000 equidistant points on the circle $S^1$ and finding the largest interval containing $\mu_0$ where the test does not reject the null hypothesis. In this case, we find the $\alpha=0.05$ level interval is $\mathcal{C}_{0.95} = (214^\circ, 242^\circ)$. Note that a full inversion of the test yields union of $(214^\circ, 242^\circ)$ and the antipodal points at $(34^\circ, 62^\circ)$ which is the result of the isotropies of antipodal points on the circle being equal, illustrating the phenomenon in the lower part of Figure \ref{fig:randomization}, Panel (b). 
\section{Conclusion} \label{sec-conclusion}

In this paper, we have developed a general framework for testing whether the Fr\'echet mean of a distribution on a metric space equals a hypothesized value. Our approach relies on constructing isotropic randomizations that preserve key geometric properties under the null hypothesis while having power to detect deviations from it. We began by examining the case of the real line, generalizing to the circle $S^1$ and further to Riemannian manifolds, where the exponential and logarithm maps provide natural tools for constructing such randomizations. Building on these insights, we extended the methodology to general metric spaces through the introduction of admissible reflections.

A key contribution is the characterization of admissible randomizations, which ensures that the resulting test has correct size as well as power against alternatives. We have shown that for radially symmetric distributions, our test is particularly well-suited as the randomization preserves distributional properties under the null hypothesis. Our numerical experiments across different metric spaces demonstrate that the test maintains the desired level while achieving good power against alternatives, even in finite samples. The application to wind direction data illustrates the practical utility of our approach in a real-world setting where traditional Euclidean methods are not applicable.

Several directions for future work emerge from this study. First, it remains to prove that the randomization scheme is consistent under the null hypothesis, which would provide a theoretical guarantee that the test has the correct size. Furthermore, as seen in Section \ref{sec-realdata}, i.i.d.\,data might not be available and investigating the behavior of the Fréchet mean and of the test under these conditions could be of interest. Additionally, the development of optimal randomization schemes, particularly for spaces where the isotropy group is rich, presents an interesting theoretical challenge. Finally, one could explore alternative test statistics, for instance based on considering the distance $d(\mu, \hat\mu_n)$.
\section{Acknowledgement} \label{sec-ack}

This work has received funding from the European Union’s Horizon 2020 research and innovation program under the Marie Skłodowska-Curie grant agreement No 956107, "Economic Policy in Complex Environments (EPOC)".

\appendix


\bibliographystyle{plain}
\bibliography{references}

\end{document}